\documentclass[12pt]{amsart}

\usepackage[english]{babel}
\usepackage[toc,page]{appendix}
\usepackage{amsmath}
\usepackage{amsthm}
\usepackage{amssymb}
\usepackage{mathrsfs}
\usepackage{enumerate}
\usepackage{mathscinet}
\usepackage[notcite, final, notref]{showkeys}
\usepackage{dsfont}
\usepackage{bookmark}
\usepackage{hhline}
\usepackage[dvips]{color}

\newcommand{\suchthat}{\;\ifnum\currentgrouptype=16 \middle\fi|\;}

\newtheorem{theorem}{Theorem}[section]

\newtheorem{proposition}[theorem]{Proposition}
\newtheorem{corollary}[theorem]{Corollary}
\newtheorem{definition}[theorem]{Definition}

\theoremstyle{definition}
\newtheorem{remark}[theorem]{Remark}

\allowdisplaybreaks

\usepackage{tikz}
\usepackage[siunitx]{circuitikz}

\usepackage{xcolor}

\title[Distributions and the Grassmann algebra]{Distribution spaces and a new construction of stochastic processes associated with the Grassmann algebra}

\author[D. Alpay]{Daniel Alpay}
\address{(DA) 
Faculty of Mathematics, Physics, and Computation\\
Schmid College of Science and Technology\\
Chapman University\\
One University Drive\\
Orange, California 92866\\
USA}
\email{alpay@chapman.edu}

\author[I. L. Paiva]{Ismael L. Paiva }
\address{(ILP)
Schmid College of Science and Technology\\
Chapman University\\
One University Drive\\
Orange, California 92866\\
USA}
\email{depaiva@chapman.edu}

\author[D. C. Struppa]{Daniele C. Struppa}
\address{(DCS) 
Faculty of Mathematics, Physics, and Computation\\
Schmid College of Science and Technology\\
Chapman University\\
One University Drive\\
Orange, California 92866\\
USA}
\email{struppa@chapman.edu}

\begin{document}
\bibliographystyle{plain}
\begin{abstract}
We associate with the Grassmann algebra a topological algebra of distributions, which allows the study of processes analogous to the corresponding free stochastic processes with stationary increments, as well as their derivatives.
\end{abstract}
\maketitle

\noindent MCS classes: 30G35, 15A75

\date{today}
\setcounter{tocdepth}{1}
\tableofcontents


\section{Introduction}
\setcounter{equation}{0}

In the present work, we develop some aspects of the Grassmann algebra counterpart of the noncommutative Fock space \cite{MR1217253} and noncommutative stochastic distributions \cite{MR3231624}. Moreover, we study an analog of stochastic processes with stationary increments (such as the fractional Brownian motion) and their derivatives in the setting of the Grassmann algebra $\Lambda$. The processes introduced here differ from the ones discussed by Rogers in, e.g., \cite{MR925919, rogers1994stochastic, rogers2003supersymmetry}.\smallskip

We recall that $\Lambda$, also called the exterior algebra, is the algebra on a field $\mathcal{K}$ generated by a finite or countable set of elements $i_n$ not belonging to and linearly independent over $\mathcal{K}$, and satisfying
\begin{equation}
i_ni_m+i_mi_n=0, \qquad n,m=1,2,\hdots
\label{anticommutative}
\end{equation}
together with the identity element of $\mathcal{K}$. Usually, the field that is considered is the setting of the complex numbers $\mathbb{C}$. Here, we follow this choice.\smallskip

Due to its role on supersymmetry, an element of $\Lambda$ is often referred to as a supernumber. When the number of generators is finite, say $N$, we use the notation $\Lambda_N$ to evidence it. In our approach, $\Lambda=\cup_{N\in\mathbb{N}}\Lambda_N$. This differs from the way $\Lambda$ is usually treated in the literature, where formal infinite sums are considered. Here, if $z\in\Lambda$, there exists $n(z)$ such that $z\in\Lambda_{n(z)}$. We will consider closures of $\Lambda$ to study cases with an effectively infinite number of generators. The construction of generators in the finite case follows with a matrix representation. In the infinite case, it is less straightforward, but concrete realizations can be given \cite{berezin1979method,rogers1980global}.\smallskip

We note that \eqref{anticommutative} also holds for many types of hypercomplex numbers if $n\neq m$ \cite{gal2002introduction}. The difference here is that it also holds for $n=m$, so that
\[
i_n^2=0,\ldots\quad{\rm for}\quad n=1, \hdots,
\]
and in particular $\Lambda$ (and every $\Lambda_N$) has divisors of $0$.\smallskip

One can say that analysis on this setting started in 1937 with Cartan showing that the Grassmann algebra can represent the exterior algebra if one introduces the idea of derivative (and multiplication) by its generators \cite{cartan1937theorie}. In 1959, Martin considered supernumbers to study ``classical versions'' of physical functions for fermions and obtain their quantization through path integrals \cite{MR0109640,MR0109639}, an idea that would be later used by Schwinger to extend his quantum field theory to fermions \cite{schwinger1966particles}. Moreover, in 1966, Berezin independently started an extensive study of what is now known as supermathematics \cite{berezin1979method, MR914369}.\smallskip

When working with $\Lambda_N$ -- and even with $\Lambda$ to a certain extent -- only algebraic operations are involved, and there is no real problem of convergence since every element in $\Lambda$ generated only by $i_n$'s is nilpotent. On the other hand, $\Lambda$ needs to be completed for various natural problems of analytical character.\smallskip

The purpose of this work is to develop counterparts of classical notions from analysis and stochastic processes theory when taking into consideration the completion of $\Lambda$ to a Hilbert space -- denoted $\overline{\Lambda}^{(2)}$ (see Definition \ref{defdef}) -- and embedding it in a Gel\cprime fand triple
\begin{equation}
\label{gelfand_triple}
\mathfrak S_{1}\subset\overline{\Lambda}^{(2)}\subset\mathfrak S_{-1}.
\end{equation}
The space $\mathfrak S_{-1}$, as we prove, has an algebra structure of the type that was first introduced by Kondratiev \cite{MR1408433} in the setting of Hida's white noise space theory, and studied in a more generalized framework in \cite{MR3404695}. There are a number of parallels (and differences) between the present study and the works \cite{MR2610579,alp}, where the complex numbers were replaced by the commutative algebra of Kondratiev stochastic distributions -- see \cite{MR1408433} for the latter, and \cite{MR1203453,MR2444857} for Hida's white noise space theory and the associated spaces of stochastic distributions. In those cases, the underlying space, i.e., the white noise space, is the commutative Fock space, which is typically associated with bosons. To obtain the so-called full Fock space, one includes the antisymmetric Fock space, which can be associated with fermions. In this context, i.e., when considering the full Fock space, it is also possible to define a noncommutative analog of the Kondratiev space \cite{MR3231624,MR3038506}. Moreover, the same type of tools can be developed in the framework of $Q$-deformed commutation relations \cite{ji2016wick}. Here, we follow a similar approach, envisioning applications on stochastic processes and their derivatives. \smallskip

As in the works \cite{aal3,MR1408433}, and more recently in the quaternionic setting \cite{MR3615375}, a Gel\cprime fand triple together with its algebra structure allows one to consider functions from a compact metric space $E$, say $[0,1]$, into $\overline{\Lambda}^{(2)}$. Such functions may be continuous, but not differentiable, when viewed as an element of $\overline{\Lambda}^{(2)}$. On the other hand, if $f$ is seen as a $\mathfrak S_{-1}$-valued function, one has, under certain hypothesis, differentiability. The differentiability is, {\it a priori}, with respect to the strong topology of $\mathfrak S_{-1}$. However, it, in fact, happens in a Hilbert space, thanks to the assumed compactness of $E$. In particular, one can study stochastic processes and their derivatives in such spaces.\smallskip

The main results of this article are the following: First, in Section \ref{sec2}, after discussing symmetries in $\Lambda$, we endow the latter with a family of norms and obtain new inequalities on these norms -- see Theorem \ref{norm-ineq}. Then, in Section \ref{fock-space}, we introduce the Fock space that can be associated with $\overline{\Lambda}^{(2)}$, making connections with classical aspects of superanalysis, e.g., left derivatives and Berezin integrals. Next, in Section \ref{topological-algebra} we embed the Fock space into Gelfand triples given by \eqref{gelfand_triple} and prove that the product in $\mathfrak S_{-1}$ satisfies V\r{a}ge-like inequalities -- see Theorem \ref{ineq}. Finally, using those Gel\cprime fand triples, we present in Section \ref{stocder} a close counterpart of the free stochastic processes with stationary increments and their derivatives.\smallskip


\section{Symmetries and norms in $\Lambda$}
\setcounter{equation}{0}
\label{sec2}

\subsection{Grassmann algebra and supernumbers}

\begin{definition}
We denote by $\mathfrak I$ the set of $t$-uples $({a_1},\ldots, {a_t})\in\mathbb N^t$, where $t$ runs through $\mathbb N$ and $a_1<{a_2}<\cdots<{a_t}$. For $\alpha=({a_1},\ldots, {a_t})\in\mathfrak I$ we set $i_\alpha=i_{a_1}\cdots i_{a_t}$ and write an element $z\in\Lambda$ as a finite sum
\begin{equation}
\label{Grassmann-decomposition}
z=z_0+\sum_{\alpha\in\mathfrak{I}}z_\alpha i_\alpha,
\end{equation}
where the coefficients $z_0$ and $z_{a_1,\ldots,a_t}$ are complex numbers.
\end{definition}

The term that does not contain any Grassmann generator, $z_0$, is called the {\it body} of the number and is sometimes denoted by $z_B$, while $z_S=z-z_B$ is said to be the {\it soul} of the number \cite{MR1172996}. One can also give a meaning to the sum \eqref{Grassmann-decomposition} when it has an infinite number of terms, as we discuss in Section \ref{comp123}.

Sometimes it is convenient to define $i_0=1$ and ``extend'' the set $\mathfrak{I}$ to accommodate it. We will denote this new set $\mathfrak{I}_0$. Hence, a supernumber can be simply written as
\[
z = \sum_{\alpha\in\mathfrak{I}_0} z_\alpha i_\alpha.
\]

If $z=\sum_{\alpha\in\mathfrak{I}_0}z_\alpha i_\alpha$ and $w=\sum_{\beta\in\mathfrak{I}_0}w_\beta i_\beta$, their product makes sense since the sums are finite, and can be written as
\[
zw = \sum_{\alpha,\beta\in\mathfrak{I}_0}z_\alpha w_\beta i_\alpha i_\beta.
\]
Let $\alpha,\beta\in\mathfrak{I}$.
Note that $i_\alpha i_\beta=0$ when $i_\alpha$ and $i_\beta$ have a common factor $i_u$, with $u\in\mathbb N$. Moreover, when $i_\alpha i_\beta$ does not vanish, it might still not be an element of the set $\{i_\alpha : \alpha\in\mathfrak{I}\}$, since permutations might be necessary to obtain such type of element. However, because permutations only introduce powers of negative one, there exists a uniquely defined $\gamma\in\mathfrak{I}$ such that
\[
i_\alpha i_\beta = (-1)^{\sigma(\alpha,\beta)} i_\gamma,
\]
where $\sigma(\alpha,\beta)$ is the number of permutations necessary to ``build'' $\gamma$ from $\alpha$ and $\beta$. If such a relation holds, we write
\begin{equation}
\alpha\vee\beta=\gamma.
\label{vee}
\end{equation}
So $i_\alpha i_\beta = (-1)^{\sigma(\alpha,\beta)} i_{\alpha\vee\beta}$. To rewrite it in a manner that includes the possibility of $i_\alpha i_\beta=0$, we define $\alpha\vee\beta=\emptyset$ if there is no $\gamma\in\mathfrak{I}_0$ such \eqref{vee} is satisfied. Then,
\[
i_\alpha i_\beta = (-1)^{\sigma(\alpha,\beta)}\sum_{\gamma\in\mathfrak{I}_0} \delta_{\alpha\vee\beta,\gamma} i_\gamma,
\]
where $\delta_{\alpha\vee\beta,\gamma}$ is the Kronecker delta.\smallskip

\begin{remark}
We note that $\mathfrak{I}_0$ defined as above is a monoid, with identity given by $\alpha=0$.
\end{remark}

\begin{remark}
We note that $\mathfrak{I}_0$ is the counterpart of the set of indices $\ell$ considered in the case of infinitely many commuting (resp., noncommuting) variables -- see \eqref{index-commut} and \eqref{index-noncommut}.
\end{remark}

It is important to observe that $\Lambda$ is a $\mathbb{Z}_2$-graded algebra. In fact, the elements that commute with each other are of the form
\begin{equation}
z=z_0+\sum_{\substack{\alpha\in\mathfrak{I} \\ |\alpha| \ {\text even}}}z_\alpha i_\alpha,
\label{even-number}
\end{equation}
where $|\alpha|$ is the number of elements of $\alpha$. Those supernumbers are called the {\it even} supernumbers and their set is denoted by $\Lambda_{even}$. It is easy to verify that they commute with {\it every} element of $\Lambda$ and that, moreover, they form a commutative subalgebra.\smallskip

On the other hand, the elements that anticommute with each other are of the type
\begin{equation}
z=\sum_{\substack{\alpha\in\mathfrak{I} \\ |\alpha| \ {\text odd}}}z_\alpha i_\alpha.
\label{odd-number}
\end{equation}
They are known as {\it odd} supernumbers and do not form a subalgebra. In fact, it is an immediate result that the product of two odd supernumbers is an even supernumber. The set of odd supernumbers is denoted by $\Lambda_{odd}$.\smallskip

The following results are easy, but they are relevant to our discussion.\smallskip

\begin{proposition}
Let $v\in\Lambda_{odd}\subset\Lambda$. Then,
\[
v^2=0.
\]
\label{odd-square}
\end{proposition}

\begin{proof}
The proof follows easily. Let $v=\sum_{\substack{\alpha\in\mathfrak{I} \\ |\alpha| \ {\text odd}}}v_\alpha i_\alpha$
\begin{eqnarray*}
v^2 & = & \frac{1}{2}\sum_{\substack{\alpha,\beta\in\mathfrak{I} \\ |\alpha|,|\beta| \ {\text odd}}} \left(v_\alpha v_\beta i_\alpha i_\beta + v_\beta v_\alpha i_\beta i_\alpha\right) \\
    & = & \frac{1}{2}\sum_{\substack{\alpha,\beta\in\mathfrak{I} \\ |\alpha|,|\beta| \ {\text odd}}}v_\alpha v_\beta \left(i_\alpha i_\beta + i_\beta i_\alpha\right) \\
    & = & 0.
\end{eqnarray*}
\end{proof}

\begin{remark}
Even though Proposition \ref{odd-square} refers to the case where $v$ is an odd supernumber in $\Lambda$, its result is still valid when considering closures of $\Lambda$, i.e., when the set $\mathfrak{I}$ has an infinite number of elements.
\end{remark}

\begin{proposition}
Let $N\in\mathbb{N}$ and consider $N+1$ elements $z_n\in\Lambda_N$ such that ${z_n}_B=0$ for every $n\in{1,\hdots,N+1}$. Then,
\[
\prod_{n=1}^{N+1} z_n = 0.
\]
In particular,
\[
z_S^{N+1} = 0
\]
for every $z=z_B+z_S\in\Lambda_N$.
\end{proposition}

We note the following two corollaries, omitting the proof for the first one.

\begin{corollary}
Let $z\in\Lambda$ be such that $z_B=0$. Then, there exists $n=n(z)$ such that $z^{n(z)+1}=0$.
\label{nz}
\end{corollary}

\begin{corollary}
Let $z=z_B+z_S\in\Lambda$. Then, $z$ is invertible if and only if $z_B\neq0$.
\end{corollary}

\begin{proof}
On the one hand, assume $z_B\neq0$. Then,
\[
z = z_B \left( 1+ \frac{z_S}{z_B}\right).
\]
According to Corollary \ref{nz}, there exists a $n(z)$ such that $z^{n(z)+1}=0$. Then,
\[
\left(1+z_B^{-1}z_S\right)^{-1} = \sum_{k=0}^{n(z)} \left(-z_B^{-1}z_S\right)^k
\]
and
\[
z^{-1} = z_B^{-1} \sum_{k=0}^{n(z)} \left(-z_B^{-1}z_S\right)^k.
\]

On the other hand, assume $z$ is invertible and let its inverse be $w=w_B+w_S\in\Lambda$. Then, $zw=1$ and, in particular,
\[
z_Bw_B=1\Rightarrow z_B\neq 0.
\]
\end{proof}

\subsection{Symmetries}

We can define a number of conjugations of a supernumber. We start with the one characterized by
\begin{gather*}
i_n^{\dag_1} = -i_n, \quad \forall n \\
\left(z w\right)^{\dag_1} = w^{\dag_1} z^{\dag_1}.
\end{gather*}
Therefore, the conjugation is defined as
\[
z^{\dag_1}=z_0+\sum_{\alpha\in\mathfrak I}z_\alpha i_\alpha^{\dag_1}.
\]

Let
\[
\pi(\alpha) = \frac{|\alpha|(|\alpha|-1)}{2}.
\]
Then, $i_\alpha^{\dag_1} = (-1)^{|\alpha| + \pi(\alpha)} i_\alpha$ and
\[
z^{\dag_1} = z_0+\sum_{\alpha\in\mathfrak I} (-1)^{|\alpha| + \pi(\alpha)} z_\alpha i_\alpha.
\]

In general, we have
\begin{align*}
\left(z+w\right)^{\dag_1} &= z^{\dag_1} + w^{\dag_1},\\
\left(z^{\dag_1}\right)^{\dag_1} &= z. \\
\end{align*}

The next conjugation is given by the conjugation of the complex coefficients. The Grassmann generators are invariant under it, i.e.,
\[
i_\alpha^{\dag_2} = i_\alpha.
\]
Therefore,
\[
z^{\dag_2}=\overline{z_0}+\sum_{\alpha\in\mathfrak I}\overline{z_\alpha} i_\alpha,
\]
where the overline represents the usual conjugation of a complex number. In general, it holds
\begin{align*}
\left(z+w\right)^{\dag_2} &= \left(z\right)^{\dag_2} + \left(w\right)^{\dag_2}\\
\left(z^{\dag_2}\right)^{\dag_2} &= z \\
\label{z1z2}
\left(z w\right)^{\dag_2} &= z^{\dag_2} w^{\dag_2}.
\end{align*}

The next conjugation is motivated by the fact already mentioned that $\Lambda$ is a $\mathbb{Z}_2$-graded algebra, i.e., the fact that an arbitrary supernumber $z\in\Lambda$ can be written as $z=u+v$, where $u\in\Lambda_{even}$ and $v\in\Lambda_{odd}$. We, then, define the conjugation $\dag_3$ in a way that $u^{\dag_3}=u$ and $v^{\dag_3} = -v$. Hence, it holds in general
\[
z^{\dag_3} = z_0 + \sum_{\alpha\in\mathfrak I}(-1)^{|\alpha|} z_\alpha i_\alpha.
\]

Observe that
\begin{align*}
\left(z+w\right)^{\dag_3} &= \left(z\right)^{\dag_3} + \left(w\right)^{\dag_3}\\
\left(z^{\dag_3}\right)^{\dag_3} &= z \\
\left(z w\right)^{\dag_3} &= \left(z\right)^{\dag_3} \left(w\right)^{\dag_3}.
\end{align*}

We note that the aforedefined conjugations commute with each other:
\begin{eqnarray*}
\left(z^{\dag_1}\right)^{\dag_2} & = & \left(z^{\dag_2}\right)^{\dag_1}, \\
\left(z^{\dag_2}\right)^{\dag_3} & = & \left(z^{\dag_3}\right)^{\dag_2}, \\
\left(z^{\dag_3}\right)^{\dag_1} & = & \left(z^{\dag_1}\right)^{\dag_3}.
\end{eqnarray*}

Therefore, four more conjugations can be defined:
\begin{align*}
z^{\dag_4} &\equiv \left(z^{\dag_1}\right)^{\dag_2} = \overline{z_0} + \sum_{\alpha\in\mathfrak I} (-1)^{|\alpha|+\pi(\alpha)} \overline{z_\alpha} i_\alpha, \\
z^{\dag_5} &\equiv \left(z^{\dag_2}\right)^{\dag_3} = \overline{z_0}+\sum_{\alpha\in\mathfrak I} (-1)^{|\alpha|} \overline{z_\alpha} i_\alpha, \\
z^{\dag_6} &\equiv \left(z^{\dag_3}\right)^{\dag_1} = z_0 + \sum_{\alpha\in\mathfrak I}(-1)^{\pi(\alpha)} z_\alpha i_\alpha, \\
z^{\dag_7} &\equiv \left(\left(z^{\dag_1}\right)^{\dag_2}\right)^{\dag_3} = \overline{z_0}+\sum_{\alpha\in\mathfrak I} (-1)^{\pi(\alpha)} \overline{z_\alpha} i_\alpha.
\end{align*}

We call special attention to $\dag_7$ because this is the conjugation that normally appears in the literature \cite{MR1172996}. Observe that it is very similar to $\dag_2$: it can be characterized as the complex conjugation of the coefficients and
\[
i_n^{\dag_7} = i_n.
\]
However, in general $i_\alpha^{\dag_7}\neq i_\alpha$. The reason for it is that given two supernumbers $z$ and $w$,
\[
\left(zw\right)^{\dag_7} = \left(w\right)^{\dag_7} \left(z\right)^{\dag_7}.
\]

Motivated by this conjugation, a {\it real supernumber} $z_R$ is defined as a supernumber with the property $\left(z_R\right)^{\dag_7}=z_R$. Analogously, a {\it imaginary supernumber} $z_C$ is a supernumber such that $\left(z_C\right)^{\dag_7}=-z_C$. With that, it is easy to see that a supernumber $z$ can be written as $z=z_R+z_C$.

\begin{remark}
The conjugations $\dag_i$ together with the identity $I$ form a commutative group with the composition law. In fact, it is the elementary abelian group $E_8$.
\end{remark}

\begin{remark}
It is clear that $zz^{\dag_i}$ (and, in particular, the modulus $\left|z\right|_i^2=zz^{\dag_i}$ induced by $\dag_i$), with $i\in\{1,2,\hdots,7\}$, is not a real number in general. To see it, first observe that $\dag_4$, $\dag_5$, $\dag_6$, and $\dag_7$ are just compositions of conjugations. Then, let $i\in\mathbb{C}$ be the complex unit and consider the examples $z=i_1 i_2 i_3 - i_4$, $w=ii_1+i_3$, and $r=1+i_1 i_2 + i_3$, which give
\[
zz^{\dag_1}=2i_1 i_2 i_3 i_4 \quad ww^{\dag_2}=2i i_1 i_2 \quad rr^{\dag_3}=1 + 2 i_1 i_2.
\]
\end{remark}

\begin{remark}
The choices of $z$ and $w$ made in the previous remark also show that
\[
zz^{\dag_1} \neq z^{\dag_1}z \quad ww^{\dag_2} \neq w^{\dag_2} w.
\]
\end{remark}

\begin{remark}
In spite of the previous remark, $zz^{\dag_3}=z^{\dag_3}z$ for every $z\in\Lambda$. In fact, $z=u+v$ and, using Proposition \ref{odd-square},
\[
zz^{\dag_3} = \left(u+v\right)\left(u-v\right)=u^2=\left(u-v\right)\left(u+v\right)=z^{\dag_3}z.
\]
\end{remark}

\begin{remark}
It is clear that the conjugations defined here are continuous in $\Lambda_N$ with respect to its natural topology induced by $\mathbb{C}^{2^N}$.
\end{remark}


\subsection{Norms and completions}
\label{comp123}

Because the conjugations defined above induce moduli that in general are not real, we introduce now the $p$-norm of a supernumber.

\begin{definition}
Let $p\geq 1$ be a real number. The $p$-norm of a supernumber $z\in\Lambda$ is defined as
\begin{equation}
\Vert z \Vert_p = \left(\sum_{\alpha\in\mathfrak{I}_0} |z_\alpha|^p \right)^{1/p},
\label{pnorm}
\end{equation}
where $|\cdot|$ is the usual modulus of a complex number.
\end{definition}

The above definition makes sense for any real $p\geq1$. However, for the purposes of this paper, we consider only $p\in\mathbb{N}$.

\begin{theorem}
Let $z,w\in\Lambda$. If $p=1$,
\begin{equation}
\Vert zw \Vert_1 \leq \Vert z \Vert_1 \Vert w \Vert_1.
\label{1norm-ineq}
\end{equation}
If $p> 1$,
\begin{equation}
\Vert zw \Vert_p^p \leq \Vert z \Vert_1^p \Vert w \Vert_{2^{p-1}} \prod_{k=1}^{p-1} \Vert w\Vert_{2^k}
\label{pnorm-ineq}
\end{equation}
and
\begin{equation}
\Vert zw \Vert_p^p \leq \Vert w \Vert_1^p \Vert z \Vert_{2^{p-1}} \prod_{k=1}^{p-1} \Vert z\Vert_{2^k}.
\label{pnorm-ineq2}
\end{equation}
\label{norm-ineq}
\end{theorem}

\begin{proof}
The proof of \eqref{1norm-ineq} is straightforward and is, then, omitted.

To start the proof of \eqref{pnorm-ineq}, we note that
\begin{eqnarray*}
\Vert zw \Vert_p^p & = & \left\Vert\sum_{\gamma\in\mathfrak{I}_0} \sum_{\alpha\vee\beta=\gamma} (-1)^{\sigma(\alpha,\beta)} z_\alpha w_\beta i_\gamma \right\Vert_p^p \\
     & = & \sum_{\gamma\in\mathfrak{I}_0} \left|\sum_{\alpha\vee\beta=\gamma} (-1)^{\sigma(\alpha,\beta)} z_\alpha w_\beta\right|^p \\
     & \leq & \sum_{\gamma\in\mathfrak{I}_0} \sum_{k=1}^p \sum_{\alpha_k\vee\beta_k=\gamma} \left|z_{\alpha_1}\right| \cdots \left|z_{\alpha_p}\right| \left|w_{\beta_1}\right| \cdots \left|w_{\beta_p}\right| \\
     & \leq & \sum_{\alpha_1,\cdots,\alpha_p\in\mathfrak{I}_0} \left|z_{\alpha_1}\right| \cdots \left|z_{\alpha_p}\right| \sum_{k=1}^p \sum_{\substack{\gamma\in\mathfrak{I}_0; \exists \beta_k\in\mathfrak{I}_0 \\ \alpha_k\vee\beta_k=\gamma}} \left|w_{\beta_1}\right| \cdots \left|w_{\beta_p}\right|.
\end{eqnarray*}
Now, using the Cauchy-Schwarz inequality,
\begin{eqnarray*}
\sum_{k=1}^p \sum_{\substack{\gamma\in\mathfrak{I}_0; \exists \beta_k\in\mathfrak{I}_0 \\ \alpha_k\vee\beta_k=\gamma}} \left|w_{\beta_1}\right| \cdots \left|w_{\beta_p}\right| & \leq & \left(\sum_{k=1}^{p-1} \sum_{\substack{\gamma\in\mathfrak{I}_0; \exists \beta_k\in\mathfrak{I}_0 \\ \alpha_k\vee\beta_k=\gamma}} \left(\left|w_{\beta_1}\right| \cdots \left|w_{\beta_{p-1}}\right|\right)^2\right)^{1/2} \left(\sum_{\substack{\gamma\in\mathfrak{I}_0 \\ \alpha_p\vee\beta_p}}|w_{\beta_p}|^2\right)^{1/2} \\
    & \leq & \left(\sum_{k=1}^{p-2} \sum_{\substack{\gamma\in\mathfrak{I}_0; \exists \beta_k\in\mathfrak{I}_0 \\ \alpha_k\vee\beta_k=\gamma}} \left(\left|w_{\beta_1}\right| 
\cdots \hspace{-2mm}
\left|w_{\beta_{p-2}}\right|\right)^4\right)^{1/4} \times \\
     &&\hspace{5mm}\times\left(\sum_{\substack{\gamma\in\mathfrak{I}_0 \\ \alpha_{p-1}\vee\beta_{p-1}}} |w_{\beta_{p-1}}|^4\right)^{1/4} \left(\sum_{\beta_p\in\mathfrak{I}_0}|w_{\beta_p}|^2\right)^{1/2} \\
     & \leq & \left(\sum_{k=1}^{p-2} \sum_{\substack{\gamma\in\mathfrak{I}_0; \exists \beta_k\in\mathfrak{I}_0 \\ \alpha_k\vee\beta_k=\gamma}} \left(\left|w_{\beta_1}\right| \cdots \left|w_{\beta_{p-2}}\right|\right)^4\right)^{1/4} \hspace{-2mm}\left(\sum_{\beta_{p-1}} |w_{\beta_{p-1}}|^4\right)^{1/4}\hspace{-3mm} \Vert w \Vert_2 \\
     & \leq & \left(\sum_{k=1}^{p-n} \sum_{\substack{\gamma\in\mathfrak{I}_0; \exists \beta_k\in\mathfrak{I}_0 \\ \alpha_k\vee\beta_k=\gamma}} \left(\left|w_{\beta_1}\right| \cdots \left|w_{\beta_{p-n}}\right|\right)^{2^n}\right)^{1/2^n} \prod_{k=1}^{n} \Vert w \Vert_{2^k} \\
     & \leq & \left(\sum_{k=1}^2 \sum_{\substack{\gamma\in\mathfrak{I}_0; \exists \beta_k\in\mathfrak{I}_0 \\ \alpha_k\vee\beta_k=\gamma}} \left(\left|w_{\beta_1}\right| \left|w_{\beta_2}\right|\right)^{2^{p-2}}\right)^{1/2^{p-2}} \prod_{k=1}^{p-2} \Vert w \Vert_{2^k} \\
     & \leq & \Vert w\Vert_{2^{p-1}} \prod_{k=1}^{p-1} \Vert w \Vert_{2^k}.
\end{eqnarray*}
Therefore,
\begin{eqnarray*}
\Vert zw \Vert_p^p & \leq & \sum_{\alpha_1,\cdots,\alpha_p\in\mathfrak{I}_0} \left|z_{\alpha_1}\right| \cdots \left|z_{\alpha_p}\right| \Vert w\Vert_{2^{p-1}} \prod_{k=1}^{p-1} \Vert w \Vert_{2^k} \\
     & \leq & \left(\sum_{\alpha\in\mathfrak{I}_0}|z_\alpha|\right)^p \Vert w\Vert_{2^{p-1}} \prod_{k=1}^{p-1} \Vert w \Vert_{2^k} \\
     & \leq & \Vert z\Vert_1^p \Vert w\Vert_{2^{p-1}} \prod_{k=1}^{p-1} \Vert w \Vert_{2^k},
\end{eqnarray*}
proving \eqref{pnorm-ineq}. The proof of \eqref{pnorm-ineq2} follows in a similar manner.
\end{proof}

\begin{remark}
The inequalities presented in Theorem \ref{norm-ineq} are homogeneous.
\end{remark}

\begin{remark}
The inequalities \eqref{pnorm-ineq} and \eqref{pnorm-ineq2} are not enough to state that, in the closure $\overline{\Lambda}^{(p)}$ of $\Lambda$ with respect to the $p$-norm, the product of elements is a law of composition, i.e., if $z,w\in\overline{\Lambda}^{(p)}$, the aforementioned inequalities do not guarantee that the product converges in $\overline{\Lambda}^{(p)}$. In this paper, we are particular interested in $\overline{\Lambda}^{(2)}$. However, we were unable to show whether the product is or is not a law of composition in this space. Although the solution to this question is important and can reveal interesting aspects of $\overline{\Lambda}^{(2)}$, the main results we present are independent of it.
\end{remark}

\begin{remark}
In \cite{MR1099318}, a Fr\'echet structure modeled on a therein defined sequence space is given to $\Lambda$ by endowing it with seminorms.
\end{remark}

Up to now, besides the last remarks, we only considered algebraic properties of $\Lambda$. To study convergence problems, it is necessary to complete this set. The completion with respect to the $1$-norm, i.e., the space $\overline{\Lambda}^{(1)}$ has a Banach algebra structure and is already known in the literature. It was introduced by Rogers \cite{rogers1980global, zbMATH00861741}. In the next section, we study $\overline{\Lambda}^{(2)}$.


\section{The Fock space}
\setcounter{equation}{0}
\label{fock-space}

The classical Fock space associated with $\ell_2$, i.e., the reproducing kernel
Hilbert space introduced by Bargmann with reproducing kernel \cite{bargmann}
\begin{equation}
\label{barg}
e^{\langle z,w\rangle_{\ell_2}}=\sum_{\alpha\in\ell}\frac{z^\alpha \overline{w}^{\alpha}}{\alpha!}
\end{equation}
corresponds to function theory in a (countably) infinite number of commuting 
complex variables. In \eqref{barg} we have used the several complex variables notation and set $z=(z_1,z_2,\ldots)\in\ell_2(\mathbb N)$, $\ell$ to be the family of sequences
\begin{equation}
\alpha=(\alpha_1,\alpha_2,\ldots), \qquad \alpha_j\in\mathbb{N}_0
\label{index-commut}
\end{equation}
where at most a finite number of $\alpha_j$ are different from $0$,
\[
z^\alpha=z_1^{\alpha_1}z_2^{\alpha_2}\cdots,
\]
and $\alpha!=\alpha_1!\alpha_2!\cdots$.

In the noncommutative setting, the commuting variables give place to noncommuting ones and, then, one needs a different set of indexes $\widetilde{\ell}$. We consider $\alpha\in\widetilde{\ell}$ given by
\begin{equation}
\alpha = ((\alpha_1,n_1),(\alpha_2,n_2),\hdots,(\alpha_m,n_m)),
\label{index-noncommut}
\end{equation}
where $n_m\in\mathbb{N}$ and $\alpha_u\neq\alpha_{u+1}$, for $u=1,\hdots,m-1$. Hence, the new kernel is written as
\[
\sum_{\alpha\in\widetilde{\ell}}z^\alpha \overline{w^\alpha},
\]
where here $\overline{w^\alpha}=\cdots \overline{w_2}^{n_2} \overline{w_1}^{n_1}$.

Note that the left concatenation gives $\widetilde{\ell}$ a monoid structure. Furthermore, it defines a partial order as follows: For $\alpha, \beta\in\widetilde{\ell}$ we say that $\beta\le \alpha$ if there is $\gamma\in\widetilde{\ell}$ such that $\alpha=\beta\gamma$.\smallskip

The parallel with the structure of the Grassmann algebra is clear. Motivated by it, we give an inner product to $\overline{\Lambda}^{(2)}$ envisioning the construction of the counterpart of the Fock space here.

\begin{definition}
The inner product $\langle \cdot,\cdot \rangle$ between two supernumbers $z,w\in\Lambda$ is defined as
\[
\langle z,w\rangle = \sum_{\alpha\in\mathfrak{I}_0} z_\alpha \overline{w_\alpha}.
\]
\end{definition}

\begin{remark}
Let $z$ and $w$ be two supernumbers. Then,
\[
\langle w,z\rangle = \langle z^{\dag_2},w^{\dag_2}\rangle = \langle z^{\dag_4},w^{\dag_4}\rangle = \langle z^{\dag_5},w^{\dag_5}\rangle = \langle z^{\dag_7},w^{\dag_7}\rangle.
\]
\end{remark}

\begin{proposition}
$\Lambda$ endowed with the aforementioned defined inner product and the $2$-norm is a Hilbert space.
\end{proposition}

\begin{proof}
It suffices to observe that the $2$-norm is induced by the above defined inner product. In fact,
\[
\Vert z\Vert_2 \equiv \langle z, z\rangle^{1/2}.
\]
\end{proof}

\begin{definition}
\label{defdef}
By analogy with the noncommutative setting, as discussed in the beginning of this section (and also noting the definition in \cite{MR1217253}), $\overline{\Lambda}^{(2)}$ is called the Fock space.
\end{definition}

In general, we consider functions $f:\mathbb{\mathcal{I}}\rightarrow\overline{\Lambda}^{(2)}$, where $\mathcal{I}$ is the domain of $f$, usually $\mathbb{C}$ or $\mathbb{R}$. Whenever we write $f\in\overline{\Lambda}^{(2)}$, we mean that $f(x)\in\overline{\Lambda}^{(2)}$ for every $x\in\mathcal{I}$.\smallskip

As in the classical examples of Fock space, it is possible to define the left multiplication operation $M_f$ in $\overline{\Lambda}^{(2)}$. But there is a caveat: as already discussed, the multiplication might not be a law of composition in $\overline{\Lambda}^{(2)}$. If that is the case, $M_f$ is unbounded for an arbitrary $f\in\overline{\Lambda}^{(2)}$. One can, then, use two different approaches: study such operators in a space of stochastic distributions $\mathfrak{S}_{-1}$, as we discuss in the next section, or restrict $M_f$ to $f\in\overline{\Lambda}^{(1)}\subset\overline{\Lambda}^{(2)}$ since Theorem \ref{norm-ineq} assures $M_f$ is bounded in this case, i.e., if $f\in\overline{\Lambda}^{(1)}\subset\overline{\Lambda}^{(2)}$ and $g\in\overline{\Lambda}^{(2)}$, we have
\[
M_f g = fg \in\overline{\Lambda}^{(2)}.
\]

In this section, we consider only $M_f$ with $f\in\overline{\Lambda}^{(1)}$. When we get back to it in Section \ref{stocder}, we consider the general case where $f\in\overline{\Lambda}^{(2)}$ since the product is a law of composition in the space of stochastic distributions, which is introduced in the next section.

Note that, if $f=\sum_{\alpha\in\mathfrak{I}_0}f_\alpha i_\alpha$,
\begin{equation}
M_f = \sum_{\alpha\in\mathfrak{I}_0}f_\alpha M_{i_\alpha}.
\end{equation}
Therefore, we can focus our analysis on multiplication by elements of $\mathfrak{I}_0$. Furthermore, we note that the focus can be on the set $\mathfrak{I}$, since multiplication by the generator of the body, i.e., $M_1$ is just the identity operator.

In the case of elements associated with $\mathfrak{I}_0$, we already know that
\[
i_\alpha i_\beta = (-1)^{\sigma(\alpha,\beta)} \sum_{\eta\in\mathfrak{I}} \delta_{\alpha\vee\beta,\eta} i_\eta.
\]
It is also straightforward from the definition of the inner product that
\[
\langle i_\alpha, i_\beta\rangle = \delta_{\alpha,\beta}.
\]
Then,
\begin{equation}
\langle M_{i_\alpha} i_\beta, i_\gamma \rangle = (-1)^{\sigma(\alpha,\beta)} \sum_{\eta\in\mathfrak{I}} \delta_{\alpha\vee\beta,\eta} \langle i_\eta, i_\gamma \rangle = (-1)^{\sigma(\alpha,\beta)} \delta_{\alpha\vee\beta,\gamma}.
\end{equation}

We can also look for an expression for the adjoint $M^*_{i_\alpha}$ of $M_{i_\alpha}$:
\begin{equation}
\langle M^*_{i_\alpha} i_\beta, i_\gamma\rangle = \langle i_\beta, M_{i_\alpha}i_\gamma\rangle = (-1)^{\sigma(\alpha,\gamma)} \delta_{\beta,\alpha\vee\gamma}.
\end{equation}
By analogy with the classical cases, we say $M^*_{i_\alpha}$ is the left derivative with respect to $i_\alpha$.

Observing that 
\[
M_{i_\alpha} = M_{i_{a_1}} M_{i_{a_2}} \hdots M_{i_{a_t}}
\]
and, as a consequence,
\begin{equation}
M_{i_\alpha}^* = M_{i_{a_t}}^* M_{i_{a_{t-1}}}^* \hdots M_{i_{a_1}}^*,
\label{left-deriv}
\end{equation}
we can pay close attention to the left derivative with respect to single generators. Then,
\[
M^*_{i_n} i_\alpha = \left\{
  \begin{array}{l l}
    0, & \text{if } a_k\neq n, \forall a_k\in\{1,\cdots,|\alpha|\} \\
    (-1)^{k-1} i_{a_1}i_{a_2}\cdots i_{a_{k-1}} i_{a_{k+1}}\cdots i_{a_{|\alpha|}}, & \text{if } \exists k; a_k=n
  \end{array}
\right.
\]
Therefore, the left derivative constructed here corresponds to the one that is traditionally defined in superanalysis and its applications \cite{berezin1979method, MR914369, MR1172996, zbMATH00861741}.

\begin{remark}
Expression \eqref{left-deriv} clearly shows that the operator conjugation $*$ is an extension of the conjugation of supernumbers $\dag_7$.
\end{remark}

\begin{remark}
The Berezin integral is a concept widely used in superanalysis and supersymmetry \cite{berezin1979method, MR914369, MR1172996, zbMATH00861741}. It coincides with the left derivative and, then, can be also defined in terms of $M^*_{i_n}$:
\[
\int di_{n} f \equiv M_{i_n}f.
\]
More generally, if $|\alpha|<\infty$,
\[
\int di_{\alpha} f \equiv \int di_{a_{|\alpha|}} \cdots\int i_{a_1} f = M^*_{i_{a_{|\alpha|}}}\cdots M^*_{i_{a_1}} f.
\]
Moreover, if $f$ is generated by $i_{a_1},\hdots,i_{a_N}$ and $i_\alpha=i_{a_1}\hdots i_{a_N}$, the Berezin integral
\[
\int di_{\alpha} f = f_{1,2,\hdots,N}
\]
reduces to
\[
\int di_{\alpha} f = \langle M_f 1, i_{\alpha}\rangle,
\]
which has some resemblance to a residue.
\end{remark}

It is also worth considering the self-adjoint operator
\begin{equation}
T_f = M_f + M_f^*.
\label{operator}
\end{equation}

\begin{proposition}
Let $f,g\in\overline{\Lambda}^{(2)}$ with $f_B=g_B=0$, then the following holds:
\[
\langle T_f 1, T_g 1\rangle = \langle f, g\rangle.
\]
\label{ipr}
\end{proposition}

\begin{proof}
First, note that $M_f^* 1 = f_B$ and $M_g^* 1 = g_B$. Then, if $f_B=g_B=0$,
\[
\langle T_f 1, T_g 1\rangle = \langle f + f_B, g + g_B\rangle = \langle f, g \rangle.
\]
\end{proof}

Such an operator $T_f$ is used later in this paper when we define stochastic processes associated with the Fock space.


\section{A topological algebra associated with $\Lambda$}
\setcounter{equation}{0}
\label{topological-algebra}

Let $\mathscr S$ denote the space of Schwartz functions (the space of test functions), and let  $\mathscr S^\prime$  be its dual (the space of tempered distributions). The Gel\cprime fand triple $(\mathscr S,\mathbf L_2(\mathbb R,dx),\mathscr S^\prime)$ plays an important role in classical analysis, see for instance \cite{MR0209834}. Gel\cprime fand triples are also defined in Hida's white noise space theory \cite{MR1408433,MR1387829} and in its noncommutative counterpart \cite{MR3231624,MR3038506}, and are used to solve stochastic differential equations and model stochastic processes and their derivatives. In this section we define Gel\cprime fand triples in the Grassmann setting, with the aim of solving similar problems.\\

One of the reasons commonly used for introducing a Gel\cprime fand triple is the fact that in a Hilbert space, one can define products different from the inner product between two elements and, often, those products are not a law of composition on such a space, i.e., the result of the product of two elements does not necessarily belong to the Hilbert space. For instance, in the white noise space, the Wick product is not a law of composition. Then, one embeds the white noise space into a space of stochastic distributions to make it a law of composition. Various choices are possible to do so.

In our case, as already discussed, it is not yet clear if the product is a law of composition in $\overline{\Lambda}^{(2)}$. So one could question if it is necessary to make use of an analogous embedding. However, making the product a law of composition is not the only reason to introduce the space of stochastic distributions. In fact, when considering stochastic processes, such spaces are necessary for the study of their derivatives. Because of it, we introduce in this section the analogous of the space of stochastic distributions. Before doing so, we review a few facts from the classical case as well as from the theory of perfect spaces and strong algebras. We refer the reader to \cite{MR0435834, GS2_english} for more information on these spaces.\smallskip

Our starting point is a decreasing family of Hilbert spaces $(\mathcal H_p,\|\cdot\|_{\mathcal{H}_p})_{p\in\mathbb Z}$, with increasing norms. The intersection $\mathcal F=\cap_{p=0}^\infty \mathcal H_p$ is a Fr\'echet space, which we assume perfect, meaning that compactness is equivalent to being bounded and compact. This will happen in particular when for every $p$ there exists $q>p$ such that the injection map from ${\mathcal H_q}$ into ${\mathcal H_p}$ is compact. An important instance is when the injection is nuclear. We identify ${\mathcal H}_p^\prime$ with ${\mathcal H}_{-p}$. Our main interest is the dual ${\mathcal F}^\prime
=\cup_{p=0}^\infty{\mathcal H}_{-p}$, which together with $\mathcal F$ and $\mathcal H_0$, for a Gel\cprime fand triple.\\

We endow the dual $\mathcal F^\prime$ with the strong topology, defined in terms of the bounded sets of $\mathcal F$. The space $\mathcal F^\prime$ is then locally convex, and the strong topology coincides with the inductive limit topology. See \cite[Section 3]{MR3029153} for a discussion.\\

Therefore, the analysis is, {\sl a priori}, done in a larger space of distributions, which is a (non-metrizable) inductive limit of Hilbert spaces (distributions here being understood as continuous functionals on the space $\mathcal F^\prime$). However, is in fact done locally in a Hilbert space. There are two reasons why this happens. The first reason is that the space of distributions is the dual of a Fr\'echet nuclear space (dual of a perfect space would suffice). The second reason is the algebra structure of $\mathfrak{S}_{-1}$, and V\r{a}ge inequality \eqref{vage_ineq}. We now state two of the main results related to such spaces. 
They are used in the proofs of Theorems \ref{stochasticintegral} and \ref{thm55}.


\begin{proposition}
\label{compact}
A set is (weakly or strongly) compact in $\mathcal F^\prime$ if and only if it is compact in one of the spaces ${\mathcal H}_{-p}$ in the corresponding norm.
\end{proposition}

\begin{proposition}
Assume $\mathcal F^\prime$ perfect. Then, weak and strong convergence of sequences are equivalent, and a sequence converges (weakly or strongly) if and only if it converges in one of the spaces ${\mathcal H}_{-p}$ in the corresponding norm.
\label{convergence}
\end{proposition}

A topological algebra is assumed to be separately continuous in each variable. It is not a trivial fact that a strong algebra is in fact jointly continuous in the two variables see \cite[IV.26, Theorem 2]{Bourbaki81EVT} and 
the discussion in \cite[pp. 215-216]{MR3404695}.\\

In the case we are interested, we define
\[
\mathcal{H}_{-p}(c_\alpha) = \left\{f = \sum_{\alpha\in\mathfrak{I}_0} f_\alpha i_\alpha \in\Lambda^{(2)} \suchthat \sum_{\alpha\in\mathfrak{I}_0} |f_\alpha|^2 c_\alpha^{-2p}<\infty \right\},
\]
where $p\in\mathbb{Z}$ and the coefficients $c_\alpha$'s form a sequence of positive real numbers such that
\begin{equation}
c_\alpha c_\beta \leq c_\gamma
\label{cond1}
\end{equation}
if $\alpha\vee\beta=\gamma$ and
\begin{equation}
\sum_{\alpha\in\mathfrak{I}_0} c_\alpha^{-2d} < \infty,
\label{cond2}
\end{equation}
where $d$ is a positive integer. Observe that
\[
\mathcal{H}_{-q}(c_\alpha) \subseteq \mathcal{H}_{-p}(c_\alpha)
\]
if $p\geq q$.

Henceforth, we denote $\mathcal{H}_{-p}(c_\alpha)$ simply by $\mathcal{H}_{-p}$.

\begin{definition}
The norm $\left\Vert f\right\Vert_{\mathcal{H}_{-p}}$ of $f\in\mathcal{H}_{-p}$ is defined as
\[
\left\Vert f\right\Vert_{\mathcal{H}_{-p}} \equiv \sum_{\alpha\in\mathfrak{I}_0} |f_\alpha|^2 c_\alpha^{-2p}.
\]
\end{definition}

\begin{proposition}
If $c_{\alpha\vee\beta}=c_\alpha c_\beta$, then $c_0=1$.
\label{c0}
\end{proposition}

\begin{proof}
On the one hand, if $c_0>1$, the inequality \eqref{cond1} does not hold in general since
\[
c_0 c_\alpha > c_\alpha = c_{0\vee\alpha}.
\]
On the other hand, if $c_0<1$, the condition \eqref{cond2} is not satisfied. In fact, $c_0^{-2d}<c_\alpha^{-2d}$ for every $\alpha\in\mathfrak{I}$ and, then, the sum in \eqref{cond2} diverges.
\end{proof}

\begin{proposition}
Let $f\in\mathcal{H}_{-p}$ with $c_\alpha>1$ if $\alpha\neq0$ and $c_0=1$. Then,
\[
\lim_{p\rightarrow\infty}\left\Vert f\right\Vert_{\mathcal{H}_{-p}} = |f_0|^2.
\]
\label{limit-norm}
\end{proposition}

\begin{proof}
Because $\lim_{p\rightarrow\infty} c_\alpha^{-2p}=0$ for every $\alpha\neq 0$,
\begin{eqnarray*}
\lim_{p\rightarrow\infty}\left\Vert f\right\Vert_{\mathcal{H}_{-p}} & = & \lim_{p\rightarrow\infty}\sum_{\alpha\in\mathfrak{I}_0} |f_\alpha|^2 c_\alpha^{-2p} \\
     & = & \sum_{\alpha\in\mathfrak{I}_0} |f_\alpha|^2 \lim_{p\rightarrow\infty} c_\alpha^{-2p} \\
     & = & |f_0|^2.
\end{eqnarray*}
\end{proof}

The following theorem introduces a V\r{a}ge-like inequality, which is the analogous of a result due to V\r{a}ge \cite{vage96} and allows the analysis of stochastic processes to be done locally in a Hilbert space.

\begin{theorem}
If $f\in\mathcal{H}_{-q}$ and $g\in\mathcal{H}_{-p}$, with $p> q$, then
\begin{equation}
\label{vage_ineq}
\left\Vert fg \right\Vert_{\mathcal{H}_{-p}} \leq C_{p-q} \left\Vert f\right\Vert_{\mathcal{H}_{-q}} \left\Vert g\right\Vert_{\mathcal{H}_{-p}},
\end{equation}
where $C_{p-q}$ is a positive constant.
\label{ineq}
\end{theorem}

\begin{proof}
Let $f\in\mathcal{H}_{-q}$ and $g\in\mathcal{H}_{-p}$. Hence, using the Cauchy-Schwarz inequality,
\begin{eqnarray*}
\left\Vert fg \right\Vert_{\mathcal{H}_{-p}}^2 & = & \sum_{\gamma\in\mathfrak{I}_0} |(fg)_\gamma|^2 c_\gamma^{-2p} \\
     & = & \sum_{\gamma\in\mathfrak{I}_0} \left|\sum_{\alpha\vee\beta=\gamma}(-1)^{\sigma(\alpha,\beta)}f_\alpha g_\beta\right|^2 c_\gamma^{-2p} \\
     & \leq & \sum_{\gamma\in\mathfrak{I}_0} \left(\sum_{\substack{\alpha\vee\beta=\gamma \\ \alpha'\vee\beta'=\gamma}}|f_\alpha| |g_\beta| |f_{\alpha'}| |g_{\beta'}| \right) c_\gamma^{-2p} \\
     & \leq & \sum_{\gamma\in\mathfrak{I}_0} \left(\sum_{\substack{\alpha\vee\beta=\gamma \\ \alpha'\vee\beta'=\gamma}}|f_\alpha| c_\alpha^{-p} \ |g_\beta| c_\beta^{-p} \ |f_{\alpha'}| c_{\alpha'}^{-p} \ |g_{\beta'}| c_{\beta'}^{-p} \right) \\
     & \leq & \sum_{\alpha,\alpha'\in\mathfrak{I}_0} |f_\alpha| c_\alpha^{-p} \ |f_{\alpha'}| c_{\alpha'}^{-p} \left( \sum_{\substack{\gamma\in\mathfrak{I}_0; \exists \beta,\beta' \\ \alpha\vee\beta=\gamma \\ \alpha'\vee\beta'=\gamma}} |g_\beta| c_\beta^{-p} \ |g_{\beta'}| c_{\beta'}^{-p} \right) \\
     & \leq & \sum_{\alpha,\alpha'\in\mathfrak{I}_0} |f_\alpha| c_\alpha^{-p} \ |f_{\alpha'}| c_{\alpha'}^{-p} \left( \sum_{\substack{\gamma\in\mathfrak{I}_0; \exists \beta \\ \alpha\vee\beta=\gamma}} |g_\beta|^2 c_\beta^{-2p}\right)^{1/2} \left(\sum_{\substack{\gamma\in\mathfrak{I}_0; \exists \beta' \\ \alpha'\vee\beta'=\gamma}} |g_{\beta'}|^2 c_{\beta'}^{-2p} \right)^{1/2} \\
     & \leq & \sum_{\alpha,\alpha'\in\mathfrak{I}_0} |f_\alpha| c_\alpha^{-p} \ |f_{\alpha'}| c_{\alpha'}^{-p} \left( \sum_{\beta\in\mathfrak{I}_0} |g_\beta|^2 c_\beta^{-2p}\right)^{1/2} \left(\sum_{\beta'\in\mathfrak{I}_0} |g_{\beta'}|^2 c_{\beta'}^{-2p} \right)^{1/2} \\
     & \leq & \left( \sum_{\alpha\in\mathfrak{I}_0} |f_\alpha| c_\alpha^{-p} \right)^2 \left\Vert g\right\Vert_{\mathcal{H}_{-p}}^2 \\
     & \leq & \left( \sum_{\alpha\in\mathfrak{I}_0} |f_\alpha| c_\alpha^{-q} \ c_\alpha^{q-p} \right)^2 \left\Vert g\right\Vert_{\mathcal{H}_{-p}}^2 \\
     & \leq & \left(\sum_{\alpha\in\mathfrak{I}_0} c_\alpha^{-2(p-q)} \right) \left\Vert f\right\Vert_{\mathcal{H}_{-q}}^2 \left\Vert g\right\Vert_{\mathcal{H}_{-p}}^2.
\end{eqnarray*}

Now, it remains to be shown that there exist $c_\alpha$'s such that $\sum_{\alpha\in\mathfrak{I}_0} c_\alpha^{-2(p-q)}<\infty$. We recall that $\alpha\in\mathfrak{I}$ means that $\alpha=(a_1,\hdots,a_n)$ for some integer $n\geq 1$, where $a_1,\hdots,a_n$ are positive integers such that $a_1< a_2 < \cdots < a_n$. The ``extension'' $\mathfrak{I}_0$ adds one more element element to $\mathfrak{I}$, namely $\alpha=0$. With that in mind, one possible family of weights $c_\alpha$ that satisfy \eqref{cond1} and \eqref{cond2} is given by
\[
c_\alpha = e^{\sum_{k=1}^n \varphi(a_k)}
\]
for every $\alpha\in\mathfrak{I}_0$, where $\varphi$ is a monotonically increasing real power series -- or at least the values $\varphi(n)$, for non-negative integers $n$, form a monotonically increasing sequence -- with $\varphi(0)=0$. We also require the coefficients of the power series to be bigger than $\ln 2^{1/2(p-q)}$.

Then, it is straightforward that
\[
c_\alpha c_\beta = c_\gamma
\]
if $\alpha\vee\beta=\gamma$ and $c_0=1$, as required by Proposition \ref{c0}. Moreover, if $d=p-q>0$
\[
\sum_{\alpha\in\mathfrak{I}_0} c_\alpha^{2(q-p)} = 1+  \sum_{n=1}^\infty \sum_{\substack{\alpha\in\mathfrak{I}; \\ |\alpha|=n}} e^{-2d\sum_{k=1}^n \varphi(a_k)}.
\]
Note that there exists $\xi>\ln 2^{1/2d}$ such that
\[
\sum_{\substack{\alpha\in\mathfrak{I}; \\ |\alpha|=1}} e^{-2d\varphi(a_1)} = \sum_{a_1=1}^\infty e^{-2d\varphi(a_1)} \leq \sum_{a_1=1}^\infty e^{-2d\xi a_1} = \frac{1}{e^{2d\xi}-1}<1.
\]
Also,
\[
\sum_{\substack{\alpha\in\mathfrak{I}; \\ |\alpha|=2}} e^{-2d (\varphi(a_1)+\varphi(a_2))} \leq \left(\sum_{a_1=1}^\infty e^{-2d\varphi(a_1)}\right)\left(\sum_{a_2=1}^\infty e^{-2d\varphi(a_2)}\right) \leq \left(\frac{1}{e^{2d\xi}-1}\right)^2
\]
and, in general,
\[
\sum_{\substack{\alpha\in\mathfrak{I}; \\ |\alpha|=n}} e^{-2d (\varphi(a_1)+\cdots+\varphi(a_n))} \leq \left(\sum_{a_1=1}^\infty e^{-2d\varphi(a_1)}\right)\cdots\left(\sum_{a_n=1}^\infty e^{-2d\varphi(a_n)}\right) \leq \left(\frac{1}{e^{2d\xi}-1}\right)^n.
\]
Hence,
\[
\sum_{\alpha\in\mathfrak{I}_0} c_\alpha^{-2d} \leq 1+  \sum_{n=1}^\infty \left(\frac{1}{e^{2d\xi}-1}\right)^n = 1+\frac{1}{e^{2d\xi}-2}.
\]
\end{proof}

We present the following corollary, whose proof is analogous to the one just presented for the previous theorem.

\begin{corollary}
If $f\in\mathcal{H}_{-p}$ and $g\in\mathcal{H}_{-q}$, with $p> q$, then
\begin{equation}
\left\Vert fg \right\Vert_{\mathcal{H}_{-p}} \leq C_{p-q} \left\Vert f\right\Vert_{\mathcal{H}_{-p}} \left\Vert g\right\Vert_{\mathcal{H}_{-q}},
\end{equation}
where $C_{p-q}$ is a positive constant.
\label{ineq2}
\end{corollary}

\begin{definition}
We define the space
\[\mathfrak{S}_1 = \cap_{p\in\mathbb{Z}} \mathcal{H}_p\]
and its topological dual
\[\mathfrak{S}_{-1}=\cup_{p\in\mathbb{Z}}, \mathcal{H}_{-p},\]
which are respectively the analogous of the space of test functions and the space of tempered distributions in the classical cases.
\end{definition}

\begin{corollary}
The space $\mathfrak{S}_{-1}$ endowed with the product is a strong algebra.
\label{strong}
\end{corollary}

\begin{proof}[Outline of the proof]
We first endow $\mathfrak{S}_{-1}$ with the inductive topology. Theorem \ref{ineq} implies that the product is separately continuous in each $\mathcal{H}_{-p}$, which is equivalent to continuity in the inductive topology. Furthermore, $\mathfrak{S}_{-1}$ inherits the associativity of the product in $\Lambda$. We have, then, a Banach algebra structure. Thus, $\mathfrak{S}_{-1}$ can be seen as the inductive limit of Banach spaces, which makes it a strong algebra. See \cite{MR3404695} for more details.
\end{proof}

\begin{remark}
Note that the inductive topology is equivalent to the strong topology.
\end{remark}

\begin{remark}
The product would also be associative in $\overline{\Lambda}^{(2)}$ if it was a law of composition there.
\end{remark}

\begin{corollary}
Let $n\in\mathbb{N}$ and $f\in\mathcal{H}_{-p}\subseteq\mathcal{H}_{-p-2}$. Then,
\[
\left\Vert f^n\right\Vert_{\mathcal{H}_{-p-2}} \leq C_2^{n-1} \left\Vert f\right\Vert_{\mathcal{H}_{-p}}^n.
\]
\label{f-conv}
\end{corollary}

\begin{proof}
First, note that $\left\Vert f\right\Vert_{\mathcal{H}_{-p-2}} \leq \left\Vert f\right\Vert_{\mathcal{H}_{-p}}$ for every $f\in\mathcal{H}_{-p}\subseteq\mathcal{H}_{-p-2}$. Hence, using Theorem \ref{ineq},
\begin{eqnarray*}
\left\Vert f^n\right\Vert_{\mathcal{H}_{-p-2}} & \leq & C_2 \left\Vert f\right\Vert_{\mathcal{H}_{-p}} \left\Vert f^{n-1}\right\Vert_{\mathcal{H}_{-p-2}} \\
     & \leq & C_2^2 \left\Vert f\right\Vert_{\mathcal{H}_{-p}}^2 \left\Vert f^{n-2}\right\Vert_{\mathcal{H}_{-p-2}} \\
     & \leq & C_2^{n-1} \left\Vert f\right\Vert_{\mathcal{H}_{-p}}^n.
\end{eqnarray*}
\end{proof}

\begin{corollary}
Let
\begin{equation}
F(\lambda)=\sum_{n\in\mathbb{N}_0} \alpha_n \lambda^n
\label{exp-ps}
\end{equation}
be an absolutely convergent power series in the open disk with radius $R$, where $\alpha_n,\lambda\in\mathbb{C}$ and $\mathbb{N}_0=\mathbb{N}\cup\{0\}$. If $f\in\mathcal{H}_{-p}$, then $F(f)$ converges in $\mathcal{H}_{-p-2}$ if
\begin{equation}
\left\Vert f\right\Vert_{\mathcal{H}_{-p}} < \frac{R}{C_2}.
\label{cond}
\end{equation}
\label{ps}
\end{corollary}

\begin{proof}
By assumption, if $|\lambda|<R$, the power series \eqref{exp-ps} converges absolutely, i.e., 
\[
\sum_{n\in\mathbb{N}_0} \left|\alpha_n \lambda^n\right|=\sum_{n\in\mathbb{N}_0} \left|\alpha_n\right| \left|\lambda^n\right|<\infty.
\]
Using Corollary \ref{f-conv}, we can study the absolutely convergence of $F(f)$ in $\mathcal{H}_{-p}$:
\begin{eqnarray*}
\sum_{n\in\mathbb{N}_0} \left\Vert\alpha_n f^n\right\Vert_{\mathcal{H}_{-p-2}} & = & \sum_{n\in\mathbb{N}_0} |\alpha_n|^2 \left\Vert f^n\right\Vert_{\mathcal{H}_{-p-2}} \\
     & \leq & \alpha_0+C_2^{-1} \sum_{n\in\mathbb{N}} |\alpha_n|^2 \left(C_2 \left\Vert f\right\Vert_{\mathcal{H}_{-p}}\right)^n.
\end{eqnarray*}
Then, $F(f)$ converges absolutely in $\mathcal{H}_{-p-2}$ if
\[
C_2 \left\Vert f\right\Vert_{\mathcal{H}_{-p}} < R \Rightarrow \left\Vert f\right\Vert_{\mathcal{H}_{-p}} < \frac{R}{C_2}.
\]
\end{proof}

\begin{corollary}
Let $F(\lambda)$ be a power series as in Corollary \ref{ps}. Then, $F(f)$ converges in $\mathfrak{S}_{-1}$ for $f\in\mathfrak{S}_{-1}$ if the body of $f$ satisfies \eqref{cond}.
\label{body-red}
\end{corollary}

\begin{proof}
If $f\in\mathfrak{S}_{-1}$, there exists an integer $q_0$ such that $f\in\mathcal{H}_{-q}$ for every $q\geq q_0$. By Theorem \ref{ps}, for $F(f)$ to converge, it is necessary that $\left\Vert f\right\Vert_{\mathcal{H}_{-q}} < R/C_2$, what does not hold in general. However, because of Proposition \ref{limit-norm}, we can reduce this condition to
\[
\left| f_0\right|^2 < \frac{R}{C_2},
\]
proving the corollary.
\end{proof}

\begin{corollary}
Let $f\in\mathfrak{S}_{-1}$. Then, $f$ is invertible if and only if its body $f_0$ satisfies $f_0\neq 0$.
\end{corollary}

\begin{proof}
On the one hand, if $g$ is the inverse of $f$ and its body is given by $g_0$, then
\[
fg=1 \Rightarrow f_0 g_0 = 1 \Rightarrow f_0\neq 0.
\]
On the other hand, if $f_0\neq 0$, consider without loss of generality $f_0=1$.  Then, Corollary \ref{body-red} implies that
\[
F(f) = \sum_{n\in\mathbb{N}_0} (1-f)^n
\]
converges if the body of $1-f$ is smaller than $C_2^{-1}$. However, $(1-f)_B=0$. Therefore, $g=F(f)\in\mathfrak{S}_{-1}$ and $g$ is the inverse of $f$.
\end{proof}


\section{Stochastic processes and their derivatives}
\setcounter{equation}{0}
\label{stocder}

In the literature, and in particular by Rogers (see, e.g., \cite{MR925919, rogers1994stochastic, rogers2003supersymmetry}), a class of $\Lambda$-valued functions of a real variable is considered and defined as stochastic processes. Here, we look at $\overline{\Lambda}^{(2)}$-valued functions, and present a different approach to stochastic processes in the Grassmannian setting. Our aim is to obtain a close counterpart of the noncommutative white noise space theory \cite{MR3231624,MR3038506}.\smallskip

Let us start reviewing this framework -- we refer the reader to \cite{MR3231624,MR1217253} for more details. First, we recall that the study of Gaussian stochastic processes can be made through the analysis of positive-definite kernels since there is a one-to-one correspondence between the two notions \cite{MR58:31324b,montreal}. In fact, the kernel coincides with the covariance of the stochastic process.

In the framework we are basing our model, as can be seen in \cite[Section 3]{aal2}, the processes are associated with positive-definite kernels of the form
\begin{equation}
K_\sigma(t,s) = \int_\mathbb{R} \frac{(e^{iut}-1)(e^{-ius}-1)}{u^2} d\sigma(u),
\label{kernel}
\end{equation}
where $\sigma$ is absolutely increasing continuous with respect to the Lebesgue measure, $d\sigma(u)=m(u)du$, such that the Stieltjes integral
\begin{equation}
\label{qwertyu}
\int_\mathbb{R}\frac{m(u)du}{u^2+1}<\infty.
\end{equation}
The reason for such choice is the fact that integrals of the form \eqref{kernel} correspond to correlation functions of zero-mean Gaussian processes with stationary increments -- see \cite{MR0012176,MR0004644}. An important example of such processes is the fractional Brownian motion, for which $d\sigma(u)=|u|^{1-2H}du$, with $H\in(1,2)$. In this case, the correlation function $K_\sigma(t,s)$ becomes
\begin{equation}
\label{h2h}
K(t,s)=\gamma_H\left(|t|^{2H}+|s|^{2H}-|t-s|^{2H}\right)
\end{equation}
where $\gamma_H$ depends only on $H$ and, if $H\neq 1/2$, is equal to
\[
\gamma_H = \frac{\cos(\pi H)\Gamma(2-2H)}{(1-2H)H},
\]
where $\Gamma$ is the Euler's Gamma function. Moreover, $\gamma_{1/2}= \pi $ by continuity.\smallskip

We, then, define an operator $S_m$ in $\mathbf{L}_2(\mathbb{R})$ such that
\begin{equation}
\widehat{S_mf}(u) = \sqrt{m(u)}\widehat{f}(u),
\label{Sm-def}
\end{equation}
where $\widehat{f}$ is the Fourier transform of $f$. Note that $S_m$ is, in general, unbounded. Its domain is
\[
\text{dom}\,S_m = \left\{f\in\mathbf{L}_2(\mathbb{R})\suchthat \int_\mathbb{R} m(u) |\widehat{f}(u)|^2 du<\infty\right\},
\]
which contains $\mathbf{1}_{[0,t]}$. Defining
\[
f_m(t) = S_m\mathbf{1}_{[0,t]}
\]
and using Plancherel's equality,
\begin{eqnarray*}
\left\langle f_m(t),f_m(s)\right\rangle_{\mathbf{L}_2(\mathbb{R})} & = & \frac{1}{2\pi} \left\langle \widehat{f}_m(t),\widehat{f}_m(s)\right\rangle_{\mathbf{L}_2(\mathbb{R})} \\
     & = & \frac{1}{2\pi} \left\langle \sqrt{m(u)}\widehat{\mathbf{1}}_{[0,t]},\sqrt{m(u)}\widehat{\mathbf{1}}_{[0,s]}\right\rangle_{\mathbf{L}_2(\mathbb{R})} \\
     & = & \frac{1}{2\pi} \left\langle m(u)\frac{e^{-iut}-1}{u},\frac{e^{-ius}-1}{u}\right\rangle_{\mathbf{L}_2(\mathbb{R})} \\
     & = & \frac{1}{2\pi} \int_\mathbb{R} \frac{(e^{iut}-1)(e^{-ius}-1)}{u^2} m(u) du. 
\end{eqnarray*}

In order to obtain the stochastic processes we are interested in, we construct a random variable associated with the functions $f_m(t)$. This is done with the introduction of the creation operator $\ell_h$, with $h\in\mathbf{L}_2(\mathbb{R})$, defined by
\[
\ell_h(f) = h\otimes f, \qquad f\in\Gamma(\mathbf{L}_2(\mathbb{R})),
\]
where $\Gamma(\mathbf{L}_2(\mathbb{R}))$ denotes the full Fock space associated with $\mathbf{L}_2(\mathbb{R})$. Finally, letting $T_h = \ell_h + \ell_h^*$, we define a random variable $X_m(t)$ as
\begin{equation}
X_m(t) \equiv T_{f_m(t)}.
\label{random-var-classical}
\end{equation}

Observe that the expected value of a random variable $X_m(t)$ can be defined by
\[
E(X_m(t)) = \left\langle \Omega,T_{f_m(t)}(\Omega)\right\rangle_\Gamma,
\]
where $\Omega$ is the vacuum state of $\Gamma$. Moreover, as expected:
\begin{equation}
E(X_m(t)X_m(s)) = \left\langle T_{f_m(t)}(\Omega),T_{f_m(s)}(\Omega)\right\rangle_\Gamma = \left\langle f_m(t),f_m(s) \right\rangle_{\mathbf{L}_2(\mathbb{R})}=K_\sigma(t,s),
\label{free-expected}
\end{equation}
where $K_\sigma(t,s)$ is given by \eqref{kernel} with $d\sigma(u)=m(u)du$, as already discussed.

The stochastic free processes associated with those variables have the concept of freeness -- in opposition to independence -- associated with them. Moreover, instead of Gaussian distributions, they have semi-circle distributions. Again, we refer the reader to \cite{MR3615375,MR1217253} for an extended discussion on this topic.

In the case of the stochastic processes we desire to define in the setting of the Grassmann numbers, we replace the operator $T_{f}$ in the expression \eqref{random-var-classical} by the operator we defined in expression \eqref{operator} with
\begin{eqnarray*}
f_m(t) & = & \sum_{n\in\mathbb{N}} \left\langle S_m \mathbf{1}_{[0,t]}, \xi_n\right\rangle_{\mathbf{L}_2(\mathbb{R})} i_n \\
     & = & \sum_{n\in\mathbb{N}} \left\langle \mathbf{1}_{[0,t]}, S_m\xi_n\right\rangle_{\mathbf{L}_2(\mathbb{R})} i_n \\
     & = & \sum_{n\in\mathbb{N}} \left(\int_0^t (S_m\xi_n)(u) du\right)i_n,
\end{eqnarray*}
where $\xi_n$ denotes the Hermite functions.

Hence,
\[
X_m(t) = \sum_{n\in\mathbb{N}} \left(\int_0^t (S_m\xi_n)(u) du\right)T_{i_n}.
\]

We can also define the expected value function $E$ in a similar manner as defined in the ``classical'' case:
\[
E(X_m(t)) = \langle 1, X_m(t) 1\rangle_{\overline{\Lambda}^{(2)}},
\]
where $1$ is the vacuum state in the Fock space introduced in Section \ref{fock-space}. Note that $E(X_m(t))=0$.

Using \eqref{ipr}, we observe that the covariance, which gives the kernel $K(t,s)$, satisfies:
\[
E(X_m(t)X_m(s)) = K(t,s) = \langle X_m(t)1, X_m(s)1\rangle_{\overline{\Lambda}^{(2)}} = \langle f_m(t), f_m(s)\rangle_{\overline{\Lambda}^{(2)}}.
\]
Hence,
\[
K(t,s) = \sum_{n\in\mathbb{N}} \left(\int_0^t(S_m\xi_n)(u)du\right)\left(\int_0^s (S_m\xi_n)(u') du'\right),
\]
which is equivalent to the kernel \eqref{kernel}.

Now, we gather a few results concerning bounds for the operator $T_f$, where $f\in\mathcal{H}_{-p}$ for some $p\in\mathbb{N}$. They are relevant in the proof of Theorem \ref{thm55}.

\begin{proposition}
For every $f\in\mathfrak{S}_{-1}$ with $\left\Vert f\right\Vert_{\mathcal{H}_{-p}}<\infty$, the operator $M_f$ is bounded from $\mathcal{H}_p$ into $\mathcal{H}_{-p}$.
\label{Mf-result}
\end{proposition}

\begin{proof}
Let $g\in\mathcal{H}_p$. Then, using Corollary \ref{ineq2}, we have
\[
\left\Vert M_f g\right\Vert_{\mathcal{H}_{-p}} = \left\Vert f g\right\Vert_{\mathcal{H}_{-p}} \leq C_{2p} \left\Vert f \right\Vert_{\mathcal{H}_{-p}} \left\Vert g\right\Vert_{\mathcal{H}_{p}}.
\]
\end{proof}

For the next result, and henceforth, we consider spaces $\mathcal{H}_{-p}$ where the coefficients $c_\alpha$ of the norm are of the type introduced in the proof of Theorem \ref{ineq}. In particular, if $\alpha\vee\beta=\gamma$, $c_\alpha c_\beta = c_\gamma$.

\begin{proposition}
For every $f\in\mathcal{H}_{-q}$, the operator $M^*_f$ is bounded from $\mathcal{H}_p$ into $\mathcal{H}_{-p}$, where $q<p$.
\end{proposition}

\begin{proof}
Let $g\in\mathcal{H}_{p}$. Then, we have
\begin{eqnarray*}
\left\Vert M^*_f g \right\Vert_{\mathcal{H}_{-p}}^2 \leq \left\Vert M^*_f g \right\Vert_{\mathcal{H}_{p}}^2 & = & \sum_{\gamma\in\mathfrak{I}_0} \left| \sum_{\alpha\vee\gamma=\beta} (-1)^{\sigma(\alpha,\gamma)}f_\alpha g_\beta\right|^2 c_\gamma^{2p} \\
     & \leq & \sum_{\gamma\in\mathfrak{I}_0} \left( \sum_{\alpha\vee\gamma=\beta} |f_\alpha| |g_\beta| \right)^2 c_\gamma^{2p} \\
     & \leq & \sum_{\alpha,\alpha'\in\mathfrak{I}_0} |f_\alpha| c_\alpha^{-p} |f_\alpha'| c_{\alpha'}^{-p} \sum_{\substack{\gamma\in\mathfrak{I}_0; \exists \beta,\beta' \\ \alpha\vee\gamma=\beta \\ \alpha'\vee\gamma=\beta'}} |g_\beta| c_\beta^{p} |g_\beta'| c_{\beta'}^{p} \\
     & \leq & \sum_{\alpha,\alpha'\in\mathfrak{I}_0} |f_\alpha| c_\alpha^{-p} |f_\alpha'| c_{\alpha'}^{-p} \left\Vert g\right\Vert_{\mathcal{H}_p}^2 \\
     & \leq & \left(\sum_{\alpha\in\mathfrak{I}_0} |f_\alpha| c_\alpha^{-q} c_\alpha^{-(p-q)} \right)^2 \left\Vert g\right\Vert_{\mathcal{H}_p}^2 \\
     & \leq & C_{p-q}^2 \left\Vert f\right\Vert_{\mathcal{H}_{-q}}^2 \left\Vert g\right\Vert_{\mathcal{H}_p}^2.
\end{eqnarray*}
\end{proof}

\begin{corollary}
For every $f\in\mathfrak{S}_{-1}$, the operator $T_f=M_f+M^*_f$ is bounded from $\mathfrak{S}_1$ into $\mathfrak{S}_{-1}$ and there exist $p>q$ such that
\[
\Vert T_f g\Vert_{-p} \leq 2C_1 \Vert f\Vert_{-q} \Vert g\Vert_p.
\]
\label{ineq-operator}
\end{corollary}

\begin{proof}
This is a direct consequence of the last two propositions. However, since Proposition \ref{Mf-result} is, in some sense, stronger than what we need, we note that if $f\in\mathcal{H}_{-q}$ and $g_\in\mathcal{H}_p$, the following holds from the last proposition and Theorem \ref{ineq}:
\begin{eqnarray*}
\Vert T_f g\Vert_{-p} & \leq & \Vert M_f g\Vert_{-p} + \Vert M^*_f g\Vert_{-p}\\
     & \leq & C_{p-q} \Vert f\Vert_{-q} \Vert g\Vert_{-p} + C_{p-q} \Vert f\Vert_{-q} \Vert g\Vert_p \\
     & \leq & 2C_{p-q} \Vert f\Vert_{-q} \Vert g\Vert_p \\
     & \leq & 2C_1 \Vert f\Vert_{-q} \Vert g\Vert_p.
\end{eqnarray*}
\end{proof}

To compute $dX_m(t)/dt$, we will make use of the spaces $\mathfrak{S}_1$ and $\mathfrak{S}_{-1}$ and assume certain growth conditions for $m$. First, we present the following 
proposition, the proof for which we refer the reader to \cite[Proposition 3.7 and Lemma 3.8]{aal3}.

\begin{proposition}
Let $m$ satisfy
\begin{equation}
m(u) \leq \left\{
\begin{array}{l l}
K |u|^{-b} & |u|\leq 1, \\
K |u|^{2N} & |u|> 1,
\end{array}
\right.
\label{m-form}
\end{equation}
where $b<2$, $N\in\mathbb{N}_0$, and $K$ is a positive real constant. Then,
\begin{equation}
\left|S_m\xi_n(t)\right| \leq D_1 n^{\frac{N+1}{2}} + D_2,
\label{S_m-ineq1}
\end{equation}
and
\begin{equation}
\left|S_m\xi_n(t)-S_m\xi_n(s)\right| \leq \left|t-s\right| \left(D_3 n^{\frac{N+2}{2}} + D_4\right),
\label{S_m-ineq2}
\end{equation}
where $D_1$, $D_2$, $D_3$, and $D_4$ are non-negative functions independent of $n$.
\end{proposition}

\begin{theorem}
\label{thm55}
Let $m$ be a positive measurable function, satisfying \eqref{m-form} and \eqref{qwertyu} (the latter for $d\sigma(t)=m(t)dt$). 
Then, for every $g\in\mathfrak{S}_1$ the function $t\mapsto X_m(t)f$ is strongly continuous in $\mathfrak{S}_{-1}$ and there exists a continuous operator $W_m(t)$ from $\mathfrak{S}_1$ into $\mathfrak{S}_{-1}$ such that
\begin{equation}
\frac{d}{dt}X_m(t)g = W_m(t)g.
\label{deriv-proc}
\end{equation}
Finally the function $t\mapsto W_m(t)g$ is continuous from $[a,b]$ into $\mathfrak{S}_{-1}$.
\end{theorem}

\begin{proof}
Since $f_m(t) = \sum_{n\in\mathbb{N}} f^m_n(t) i_n$, with $f^m_n = \int_0^t S_m\xi_n(u) du$, we use \eqref{S_m-ineq1} to obtain
\[
\left\Vert\frac{d}{dt} f_m(t) \right\Vert_{\mathcal{H}_{-p}}^2 = \sum_{n\in\mathbb{N}} \left|S_m\xi_n(t)\right|^2 c_n^{-2p} \leq \sum_{n\in\mathbb{N}} (D_1 n^{\frac{N+1}{2}}+D_2)^2 c_n^{-2p}.
\]
Then, for every suitable choice of coefficients $c_n$, there exists a positive integer $p_0$ such that, for every $p\geq p_0$, $df_m(t)/dt\in\mathcal{H}_{-p}$. Moreover, for $s=t+h$, where $h\neq0$ is a real number, there exists $p_1$ such that,
\[
K_p=\sum_{n\in\mathbb{N}}(D_3 n^{\frac{N+2}{2}}+D_4)^2 c_n^{-2p}<\infty, \qquad \forall p\geq p_1
\]
and, with \eqref{S_m-ineq2}, the $\mathcal{H}_{-p}$ norm of the difference of derivatives satisfies
\[
\sum_{n\in\mathbb{N}} \left|S_m\xi_n(s)-S_m\xi_n(t)\right|^2 c_n^{-2p} \leq  K_{p_1}|h|^2.
\]
Therefore, $W_m=T_{df_m/dt}$ is a continuous operator from $\mathfrak{S}_1$ into $\mathfrak{S}_{-1}$.

Finally, to see that \eqref{deriv-proc} holds, observe that for a real number $h\neq 0$ and $g\in\mathcal{H}_p$
\[
\left(\frac{X_m(t+h)-X_m(t)}{h}-W_m(t)\right)g = \sum_{n\in\mathbb{N}} \frac{\int_t^{t+h}(S_m\xi(u)-S_m\xi(t))du}{h} M_{i_n}g = X_{\Delta(t,h)}g,
\]
with
\[
\Delta(t,h) = \sum_{n\in\mathbb{N}} \frac{\int_t^{t+h}(S_m\xi(u)-S_m\xi(t))du}{h} i_n.
\]
Then, there exists $p>q\geq p_1$ such that, using Corollary \ref{ineq-operator},
\[
\left\Vert X_{\Delta(t,h)}g\right\Vert_{\mathcal{H}_{-p}} \leq 2C_1 \left\Vert\Delta(t,h)\right\Vert_{\mathcal{H}_{-q}} \left\Vert g\right\Vert_{\mathcal{H}_p} \leq \left(2C_1 K_{p_1}\left\Vert g\right\Vert_{\mathcal{H}_p}\right) |h|^2.
\]

Therefore, $W_m(t)=T_{df_m/dt}(t)\equiv dX_m(t)/dt$ is a continuous operator from $\mathfrak{S}_1$ into $\mathfrak{S}_{-1}$. Finally, using Proposition \ref{convergence}, one sees that the function $t\mapsto W_m(t)g$ is continuous.
\end{proof}

Once the treatment for derivatives is formalized, the next natural step would be the development of the counterpart of Ito/Malliavin stochastic calculus. The first step in this direction is the introduction of stochastic integrals. This is done in the following theorem, while stochastic calculus itself will be developed elsewhere.

\begin{theorem}
\label{stochasticintegral}
Let $t\mapsto Y(t)$, with $t\in[a,b]$, be a continuous $\mathfrak{S}_{-1}$-valued function in the strong topology of $\mathfrak{S}_{-1}$. If $g\in\mathfrak{S}_1$, there exists a positive integer $p$, which depends on $g$, such that the Pettis integral
\begin{equation}
\int_a^b Y(t) W_m(t)g \ dt
\label{integral}
\end{equation}
can be computed as a limit of Riemann sums and converges in $\mathcal{H}_{-p}$.
\end{theorem}

\begin{proof}[Outline of the proof]
As in \cite{aal3} the function $t\mapsto W_m(t)g$ is continuous from $[a,b]$ to $\mathfrak{S}_{-1}$. Since the product is jointly continuous in $\mathfrak{S}_{-1}$, the map $t\mapsto Y(t)W_m(t)g$ is continuous and its image is therefore compact in $\mathfrak{S}_{-1}$. By Proposition \ref{compact}, there exists $p$ such that the image of $t\mapsto Y(t)W_m(t)g$ is in $\mathcal{H}_{-p}$; this function is still continuous with respect to the topology of $\mathcal{H}_{-p}$, as is seen using Theorem \ref{convergence}. We then compute \eqref{integral} in $\mathcal{H}_{-p}$ using Riemann sums.
\end{proof}

As we already discussed, the random variables introduced here have the same expected value and the same covariance of the variables associated with free stochastic processes. In spite of it, the fact that the product in the integral \eqref{integral} is, ultimately, the product of $i_\alpha$'s, what contrasts with the ``classical'' case, which is a product of Hermite functions, shows that the processes induced by them are very different.


\section{Final remarks}
\setcounter{equation}{0}
\label{final-rmk}

In this paper, we have studied some aspects of the closure of the Grassmann algebra with respect to the $2$-norm, which we called the Fock space, as well as its embedding in Gel\cprime fand triples. We also introduced a new type of stochastic processes and an approach to study their derivatives. We present now research questions that emerge from the ideas discussed here. \smallskip

An important matter that is still not resolved is whether the product is a law of composition in the Fock space. We know that even if there is no problem of convergence for the product, the evaluation of stochastic integrals could be problematic if we had not introduced the space of stochastic distributions. However, addressing this question is crucial and either answer will, surely, lead to new results on the structure of the Fock space introduced here. \smallskip

Another direction that can be investigated is related to Wiener algebras. In \cite{MR3404695}, a special type of Wiener algebra which can be associated with any strong algebra was introduced in a generic framework. An open question is regarding whether the counterpart of such a Wiener algebra can reveal important aspects of the strong algebra we introduced, uncovering new aspects of it. \smallskip

The stochastic processes we used as a basis have the concept of freeness -- in opposition to independence -- associated with their random variables, and the distribution associated with them are semi-circles -- not Gaussians. A possible research direction is, then, the investigation of the questions: What is the independence-like concept associated with the random variables defined here? What are the distributions associated with the stochastic processes generated by them? Furthermore, as already mentioned at the end of the previous section, there is a possibility to develop the counterpart of Ito and Malliavin stochastic calculus in the present setting. One could also look for the physical processes modeled by them. \smallskip

Furthermore, another interesting direction is to look into generalizations of the processes defined here. We considered in Section \ref{stocder} the special cases where $d\sigma(t)=m(t)dt$. For a general $\sigma$, one cannot introduce the operator $S_m$ defined by \eqref{Sm-def}. However, it is possible to prove that there exists a continuous positive operator $A$ from the Schwartz space $\mathcal S$ into its dual $\mathcal S^\prime$ such that
\begin{equation}
\int_\mathbb{R} |\widehat{f}(u)|^2 d\sigma(u)=\langle Af,f\rangle_{\mathcal S^\prime,\mathcal S}.
\end{equation}
The operator $A$ can be factorized via a Hilbert space since $\mathcal S$ is nuclear -- see \cite{MR647140} for factorization theorems. An explicit construction of $A$ in the form $A=Q_\sigma^* Q_\sigma$, where $Q_\sigma$ is continuous from $\mathcal S$ into $\mathbf L_2(\mathbb R)$, is given in \cite{MR2793121}. \smallskip

Finally, a problem we intend to address soon in a different paper concerns the construction of the counterpart of linear systems and rational functions in the setting of the supernumbers. For the theory of linear systems in the setting of Hida's white noise theory and its noncommutative counterpart, see \cite{MR0255260, MR0452844, SSR}, and \cite{MR2610579, alp, MR3038506} for more recent works.

\section*{Acknowledgement}
Daniel Alpay thanks the Foster G. and Mary McGaw Professorship in Mathematical Sciences, which supported this research. Ismael L. Paiva acknowledges financial support from the Science without Borders program (CNPq/Brazil). Daniele C. Struppa thanks the Donald Bren Distinguished Chair in Mathematics, which supported this research.


\bibliographystyle{plain}

\begin{thebibliography}{10}

\bibitem{aal2}
D.~Alpay, H.~Attia, and D.~Levanony.
\newblock On the characteristics of a class of {G}aussian processes within the
  white noise space setting.
\newblock {\em Stochastic processes and applications}, 120:1074--1104, 2010.

\bibitem{aal3}
D.~Alpay, H.~Attia, and D.~Levanony.
\newblock White noise based stochastic calculus associated with a class of
  {G}aussian processes.
\newblock {\em Opuscula Mathematica}, 32/3:401--422, 2012.

\bibitem{MR3615375}
D.~Alpay, F.~Colombo, and I.~Sabadini.
\newblock On a class of quaternionic positive definite functions and their
  derivatives.
\newblock {\em J. Math. Phys.}, 58(3):033501, 15, 2017.

\bibitem{MR2793121}
D.~Alpay, P.~Jorgensen, and D.~Levanony.
\newblock A class of {G}aussian processes with fractional spectral measures.
\newblock {\em J. Funct. Anal.}, 261(2):507--541, 2011.

\bibitem{MR3231624}
D.~Alpay, P.~Jorgensen, and G.~Salomon.
\newblock On free stochastic processes and their derivatives.
\newblock {\em Stochastic Process. Appl.}, 124(10):3392--3411, 2014.

\bibitem{MR2610579}
D.~Alpay and D.~Levanony.
\newblock Linear stochastic systems: a white noise approach.
\newblock {\em Acta Appl. Math.}, 110(2):545--572, 2010.

\bibitem{alp}
D.~Alpay, D.~Levanony, and A.~Pinhas.
\newblock Linear stochastic state space theory in the white noise space
  setting.
\newblock {\em {SIAM} {Journal of Control and Optimization}}, 48:5009--5027,
  2010.

\bibitem{MR3038506}
D.~Alpay and G.~Salomon.
\newblock Non-commutative stochastic distributions and applications to linear
  systems theory.
\newblock {\em Stochastic Process. Appl.}, 123(6):2303--2322, 2013.

\bibitem{MR3029153}
D.~Alpay and G.~Salomon.
\newblock Topological convolution algebras.
\newblock {\em J. Funct. Anal.}, 264(9):2224--2244, 2013.

\bibitem{MR3404695}
D.~Alpay and G.~Salomon.
\newblock On algebras which are inductive limits of {B}anach spaces.
\newblock {\em Integral Equations Operator Theory}, 83(2):211--229, 2015.

\bibitem{bargmann}
V.~Bargmann.
\newblock Remarks on a {H}ilbert space of analytic functions.
\newblock {\em Proceedings of the {N}ational {A}cademy of {Arts}}, 48:199--204,
  1962.

\bibitem{berezin1979method}
F.A. Berezin.
\newblock The method of second quantization (Academic, New York, 1966).
\newblock page~52, 1979.

\bibitem{MR914369}
F.A. Berezin.
\newblock {\em Introduction to superanalysis}, volume~9 of {\em Mathematical
  Physics and Applied Mathematics}.
\newblock D. Reidel Publishing Co., Dordrecht, 1987.
\newblock Edited and with a foreword by A. A. Kirillov, With an appendix by V.
  I. Ogievetsky, Translated from the Russian by J. Niederle and R. Koteck\'y,
  Translation edited by Dimitri Le\u\i tes.

\bibitem{Bourbaki81EVT}
N.~Bourbaki.
\newblock {\em Espaces vectoriels topologiques}.
\newblock Masson, 1981.

\bibitem{cartan1937theorie}
E.~Cartan and J.~Leray.
\newblock La th{\'e}orie des groupes finis et continus et la g{\'e}om{\'e}trie
  diff{\'e}rentielle: trait{\'e}es par la m{\'e}thode du rep{\`e}re mobile.
\newblock 1937.

\bibitem{MR1172996}
B.~DeWitt.
\newblock {\em Supermanifolds}.
\newblock Cambridge Monographs on Mathematical Physics. Cambridge University
  Press, Cambridge, second edition, 1992.

\bibitem{gal2002introduction}
S.~Gal.
\newblock {\em Introduction to geometric function theory of hypercomplex
  variables}, volume~10.
\newblock Nova Publishers, 2002.

\bibitem{GS2_english}
I.M. Gel{\cprime}fand and G.E. Shilov.
\newblock {\em Generalized functions. Volume 2}.
\newblock Academic Press, 1968.

\bibitem{MR0435834}
I.M. Gel{\cprime}fand and N.Ya. Vilenkin.
\newblock {\em Generalized functions. {V}ol. 4}.
\newblock Academic Press [Harcourt Brace Jovanovich Publishers], New York, 1964
  [1977].
\newblock Applications of harmonic analysis, Translated from the Russian by
  Amiel Feinstein.

\bibitem{MR647140}
J.~G{\'o}rniak and A.~Weron.
\newblock Aronszajn-{K}olmogorov type theorems for positive definite kernels in
  locally convex spaces.
\newblock {\em Studia Math.}, 69(3):235--246, 1980/81.

\bibitem{MR1203453}
T.~Hida, H.~Kuo, H, and N.~Obata.
\newblock Transformations for white noise functionals.
\newblock {\em J. Funct. Anal.}, 111(2):259--277, 1993.

\bibitem{MR2444857}
T.~Hida and Si~Si.
\newblock {\em Lectures on white noise functionals}.
\newblock World Scientific Publishing Co. Pte. Ltd., Hackensack, NJ, 2008.

\bibitem{MR1408433}
H.~Holden, B.~{\O}ksendal, J.~Ub{\o}e, and T.~Zhang.
\newblock {\em Stochastic partial differential equations}.
\newblock Probability and its Applications. Birkh\"auser Boston Inc., Boston,
  MA, 1996.

\bibitem{MR1099318}
A.~Inoue and Y.~Maeda.
\newblock Foundations of calculus on super {E}uclidean space {${\bf R}^{m|n}$}
  based on a {F}r\'echet-{G}rassmann algebra.
\newblock {\em Kodai Math. J.}, 14(1):72--112, 1991.

\bibitem{ji2016wick}
U.C. Ji and E.~Lytvynov.
\newblock Wick calculus for noncommutative white noise corresponding to
  q-deformed commutation relations.
\newblock {\em Complex Analysis and Operator Theory}, pages 1--21, 2016.

\bibitem{MR0255260}
R.~E. Kalman, P.~L. Falb, and M.~A. Arbib.
\newblock {\em Topics in mathematical system theory}.
\newblock McGraw-Hill Book Co., New York, 1969.

\bibitem{MR0012176}
M.G. Krein.
\newblock On the problem of continuation of helical arcs in {H}ilbert space.
\newblock {\em C. R. (Doklady) Acad. Sci. URSS (N.S.)}, 45:139--142, 1944.

\bibitem{MR1387829}
H.-H. Kuo.
\newblock {\em White noise distribution theory}.
\newblock Probability and Stochastics Series. CRC Press, Boca Raton, FL, 1996.

\bibitem{MR58:31324b}
M.~Lo{\`e}ve.
\newblock {\em Probability theory. {II}}.
\newblock Springer-Verlag, New York, fourth edition, 1978.
\newblock Graduate Texts in Mathematics, Vol. 46.

\bibitem{MR0109640}
J.L. Martin.
\newblock The {F}eynman principle for a {F}ermi system.
\newblock {\em Proc. Roy. Soc. London. Ser. A}, 251:543--549, 1959.

\bibitem{MR0109639}
J.L. Martin.
\newblock Generalized classical dynamics, and the ``classical analogue'' of a
  {F}ermi oscillator.
\newblock {\em Proc. Roy. Soc. London. Ser. A}, 251:536--542, 1959.

\bibitem{MR0004644}
{J. von} Neumann and I.~J. Schoenberg.
\newblock Fourier integrals and metric geometry.
\newblock {\em Trans. Amer. Math. Soc.}, 50:226--251, 1941.

\bibitem{montreal}
J.~Neveu.
\newblock {\em Processus al\'eatoires gaussiens}.
\newblock Les presses de l'Universit\'e de Montr\'eal, 1968.

\bibitem{rogers1980global}
A.~Rogers.
\newblock A global theory of supermanifolds.
\newblock {\em Journal of Mathematical Physics}, 21(6):1352--1365, 1980.

\bibitem{MR925919}
A.~Rogers.
\newblock Fermionic path integration and {G}rassmann {B}rownian motion.
\newblock {\em Comm. Math. Phys.}, 113(3):353--368, 1987.

\bibitem{rogers1994stochastic}
A.~Rogers.
\newblock Stochastic calculus and anticommuting variables.
\newblock {\em arXiv preprint hep-th/9409162}, 1994.

\bibitem{rogers2003supersymmetry}
A.~Rogers.
\newblock Supersymmetry and brownian motion on supermanifolds.
\newblock {\em Infinite Dimensional Analysis, Quantum Probability and Related
  Topics}, 6(supp01):83--102, 2003.

\bibitem{zbMATH00861741}
A.~{Rogers}.
\newblock {\em {Supermanifolds. Theory and applications.}}
\newblock Singapore: World Scientific, 2007.

\bibitem{MR0209834}
L.~Schwartz.
\newblock {\em Th\'eorie des distributions}.
\newblock Publications de l'Institut de Math\'ematique de l'Universit\'e de
  Strasbourg, No. IX-X. Nouvelle \'edition, enti\'erement corrig\'ee, refondue
  et augment\'ee. Hermann, Paris, 1966.

\bibitem{schwinger1966particles}
J.~Schwinger.
\newblock Particles and sources.
\newblock {\em Physical Review}, 152(4):1219, 1966.

\bibitem{MR0452844}
E.D. Sontag.
\newblock On linear systems and noncommutative rings.
\newblock {\em Math. Systems Theory}, 9(4):327--344, 1975/76.

\bibitem{SSR}
E.D. Sontag.
\newblock Linear systems over commutative rings: A survey.
\newblock {\em Ricerche di Automatica}, 7:1--34, 1976.

\bibitem{vage96}
G.~V{\r{a}}ge.
\newblock Hilbert space methods applied to stochastic partial differential
  equations.
\newblock In H.~K{\"o}rezlioglu, B.~{\O}ksendal, and A.S. {\"U}st{\"u}nel,
  editors, {\em Stochastic analysis and related topics}, pages 281--294.
  Birk{\"a}user, {B}oston, 1996.

\bibitem{MR1217253}
D.V. Voiculescu, K.J. Dykema, and A.~Nica.
\newblock {\em Free random variables}, volume~1 of {\em CRM Monograph Series}.
\newblock American Mathematical Society, Providence, RI, 1992.
\newblock A noncommutative probability approach to free products with
  applications to random matrices, operator algebras and harmonic analysis on
  free groups.

\end{thebibliography}
\def\cprime{$'$} \def\cprime{$'$} \def\cprime{$'$}
  \def\lfhook#1{\setbox0=\hbox{#1}{\ooalign{\hidewidth
  \lower1.5ex\hbox{'}\hidewidth\crcr\unhbox0}}} \def\cprime{$'$}
  \def\cprime{$'$} \def\cprime{$'$} \def\cprime{$'$} \def\cprime{$'$}
  \def\cprime{$'$}

\end{document}